\documentclass{article}

\newif\ificml
\icmlfalse

\ificml

\usepackage{times}
\usepackage{graphicx} 
\usepackage{epstopdf}
\usepackage{subfigure}

\usepackage{natbib}

\usepackage{algorithm}
\usepackage{algorithmic}

\usepackage{hyperref}

\usepackage[accepted]{icml2016}

\usepackage{amsmath,amsthm,amssymb}

\else

\usepackage{amsmath,amsfonts,amsthm,amssymb}
\usepackage{fullpage}
\usepackage{url}
\usepackage{algorithm}
\usepackage{algorithmic}
\usepackage{float}
\usepackage{enumerate}
\usepackage{epsfig}
\usepackage{subfigure}
\usepackage{color}
\usepackage[usenames,dvipsnames,svgnames,table]{xcolor}
\usepackage[colorlinks=true,
            linkcolor=red,
            urlcolor=blue,
            citecolor=gray]{hyperref}

\fi

\makeatletter
\newtheorem*{rep@theorem}{\rep@title}
\newcommand{\newreptheorem}[2]{%
\newenvironment{rep#1}[1]{%
 \def\rep@title{#2 \ref{##1}}%
 \begin{rep@theorem}}%
 {\end{rep@theorem}}}
\makeatother

\newreptheorem{theorem}{Theorem}
\newreptheorem{lemma}{Lemma}
\newreptheorem{corollary}{Corollary}

\usepackage{rfmacros}
  \usepackage{nth}
  \usepackage{intcalc}

\newcommand\vA{\bv{A}}
\newcommand\vb{\bv{b}}
\newcommand\vx{\bv{x}}
\newcommand\gap{\gamma}

\DeclareMathOperator{\sr}{sr}

\newcommand\vU{\bv{U}}

\newcommand\vSigma{\bv{\Sigma}}
\newcommand\vV{\bv{V}}

\newcommand\vS{\bv{S}}
\newcommand{\algA}{\mathcal{A}}
\newcommand{\algB}{\mathcal{B}}
\newcommand{\algC}{\mathcal{C}}

\newcommand\pinv{\dagger}
\newcommand\ridgesol[2]{{#1}^{#2}}

\begin{document}

\ificml
\else

\title{Principal Component Projection \\Without Principal Component Analysis}
\author{Roy Frostig\\Stanford University\\\texttt{rf@cs.stanford.edu} \and
Cameron Musco\\MIT\\\texttt{cnmusco@mit.edu}  \and
Christopher Musco\\MIT\\\texttt{cpmusco@mit.edu}  \and 
Aaron Sidford\\Microsoft Research, New England\\\texttt{asid@microsoft.com}}
\date{February 23, 2016}
\maketitle

\fi

\begin{abstract}
We show how to efficiently project a vector onto the top principal
components of a matrix, \emph{without explicitly computing these
  components}. Specifically, we introduce an iterative algorithm that
provably computes the projection using few calls to any black-box
routine for ridge regression.

By avoiding explicit
principal component analysis (PCA), our algorithm is the first with no runtime dependence on the number of
top principal components. 
We show that it can be used to give a fast iterative method for the popular
principal component regression problem, giving the first major
runtime improvement over the naive method of combining
PCA with regression.

To achieve our results, we first observe that ridge regression can be used to obtain a ``smooth projection'' onto the top principal components. We then sharpen this approximation to true projection using a low-degree polynomial approximation to the matrix step function.
Step function approximation is a topic of long-term interest in
scientific computing. We extend prior theory by constructing
polynomials with simple iterative structure and rigorously analyzing their behavior under limited precision.
\end{abstract}

\section{Introduction}
In machine learning and statistics, it is common -- and often
essential -- to represent data in a concise form that decreases noise and increases
efficiency in downstream tasks. Perhaps the most
widespread method for doing so is to project data onto the linear subspace spanned by
its directions of highest variance -- that is, onto the span of the top
components given by principal component analysis (PCA).
\ificml

\fi
Computing principal components can be an expensive task, a challenge that prompts 
a basic algorithmic question:

\begin{quote}
\emph{Can we project a vector onto the span of a matrix's top principal components without performing principal component analysis?}
\end{quote}

This paper answers that question in the affirmative, demonstrating that projection
is much easier than PCA itself. We show that it can be solved using a simple iterative algorithm
based on black-box calls to a ridge regression routine. 
\ificml

\fi
The algorithm's runtime \emph{does not depend} on the number of top principal components chosen for projection, a cost inherent to any algorithm for PCA, or even algorithms that just compute an orthogonal span for the top components.

\subsection{Motivation: principal component regression}

To motivate our projection problem, consider one of the most basic downstream
applications for PCA: linear regression. Combined, PCA and regression
comprise the \emph{principal component regression} (PCR) problem:
\begin{definition}[Principal component regression (PCR)]
Let $\bv{A} \in \R^{n\times d}$ be a design matrix whose rows are data points and let $\bv{b} \in \R^d$ be a vector of data labels.
Let $\bv{A}_{\lambda}$ denote the result of projecting each row of $\bv{A}$ onto the span of the top principal components of $\bv{A}$ -- in particular the eigenvectors
of the covariance matrix $\frac{1}{n} \bv{A}^\T \bv{A}$ whose corresponding variance (eigenvalue) exceeds a threshold $\lambda$. The
task of PCR is to find a minimizer of the squared loss
$\ltwo{\vA_\lambda \vx - \vb}^2$. In other words, the goal is to
compute $\bv{A}_\lambda^\pinv \bv{b}$, where $\bv{A}_\lambda^\pinv$
is the Moore-Penrose pseudoinverse of $\bv{A}_\lambda$.
\end{definition}
PCR is a key regularization method in statistics,
numerical linear algebra, and scientific disciplines including chemometrics
\cite{hotelling1957relations, hansen1987truncatedsvd, Friedman}. 
It
models the assumption that small principal components represent noise
rather than data signal. 
PCR is typically solved by first using PCA to compute $\bv{A}_\lambda$ and then applying
linear regression. The PCA step dominates the algorithm's cost, especially if many principal components have
variance above the threshold $\lambda$. 

We remedy this issue by showing that our principal component \emph{projection} algorithm yields a fast algorithm for \emph{regression}. Specifically, 
full access to $\bv{A}_\lambda$ is unnecessary for PCR: $\bv{A}_\lambda^\pinv \bv{b}$ can be computed efficiently given only an approximate projection of the vector $\bv{A}^\T\bv{b}$ onto $\bv{A}$'s top principal components. By solving projection without PCA we obtain the first PCA-free algorithm for PCR.

\subsection{A first approximation: ridge regression}

Interestingly, our main approach to efficient principal component projection 
is based on a common alternative to PCR: ridge regression.
This ubiquitous regularization method computes a
minimizer of $\norm{\bv{A} \bv{x} - \vb}_2^2 + \lambda
\norm{\bv{x}}_2^2$ for some regularization parameter $\lambda$ \cite{tikhonov1963solution}. 
The advantage of ridge regression is its formulation as a simple
convex optimization problem that can be solved efficiently using many techniques (see Lemma \ref{ridgeRuntime}).

Solving ridge regression is equivalent to applying the matrix $(\bv{A}^\T\bv{A}+\lambda \bv{I})^{-1}\bv{A}^\T$, an 
operation that can be viewed as a smooth relaxation of PCR. 
Adding the $\ell_2$ norm penalty (i.e.\ $\lambda \bv{I}$) effectively ``washes out'' $\bv{A}$'s small principal components in
comparison to its large ones and achieves an
effect similar to PCR at the extreme ends of $\bv{A}$'s spectrum.

Accordingly, ridge regression gives access to a ``smooth projection''
operator, $(\bv{A}^\T\bv{A}+\lambda \bv{I})^{-1}\bv{A}^\T\bv{A}$. This
matrix approximates $\bv{P}_{\bv{A}_\lambda}$, the projection matrix
onto $\bv{A}$'s top principal components. Both have the same singular
vectors, but $\bv{P}_{\bv{A}_\lambda}$ has a singular value of $1$ for
each squared singular value $\sigma_i^2 \geq \lambda$ in $\bv{A}$ and
a singular value of $0$ for each $\sigma_i^2 < \lambda$, whereas
$(\bv{A}^\T\bv{A}+\lambda \bv{I})^{-1}\bv{A}^\T\bv{A}$ has singular
values equal to $\frac{\sigma_i^2}{\sigma_i^2 + \lambda}$. This
function approaches $1$ when $\sigma_i^2$ is much greater than
$\lambda$ and $0$ when it is smaller. Figure~\ref{soft_projection}
illustrates the comparison.

\begin{figure}[ht]
\begin{center}
\ificml
\else
\centerline{\includegraphics[width=.5\columnwidth]{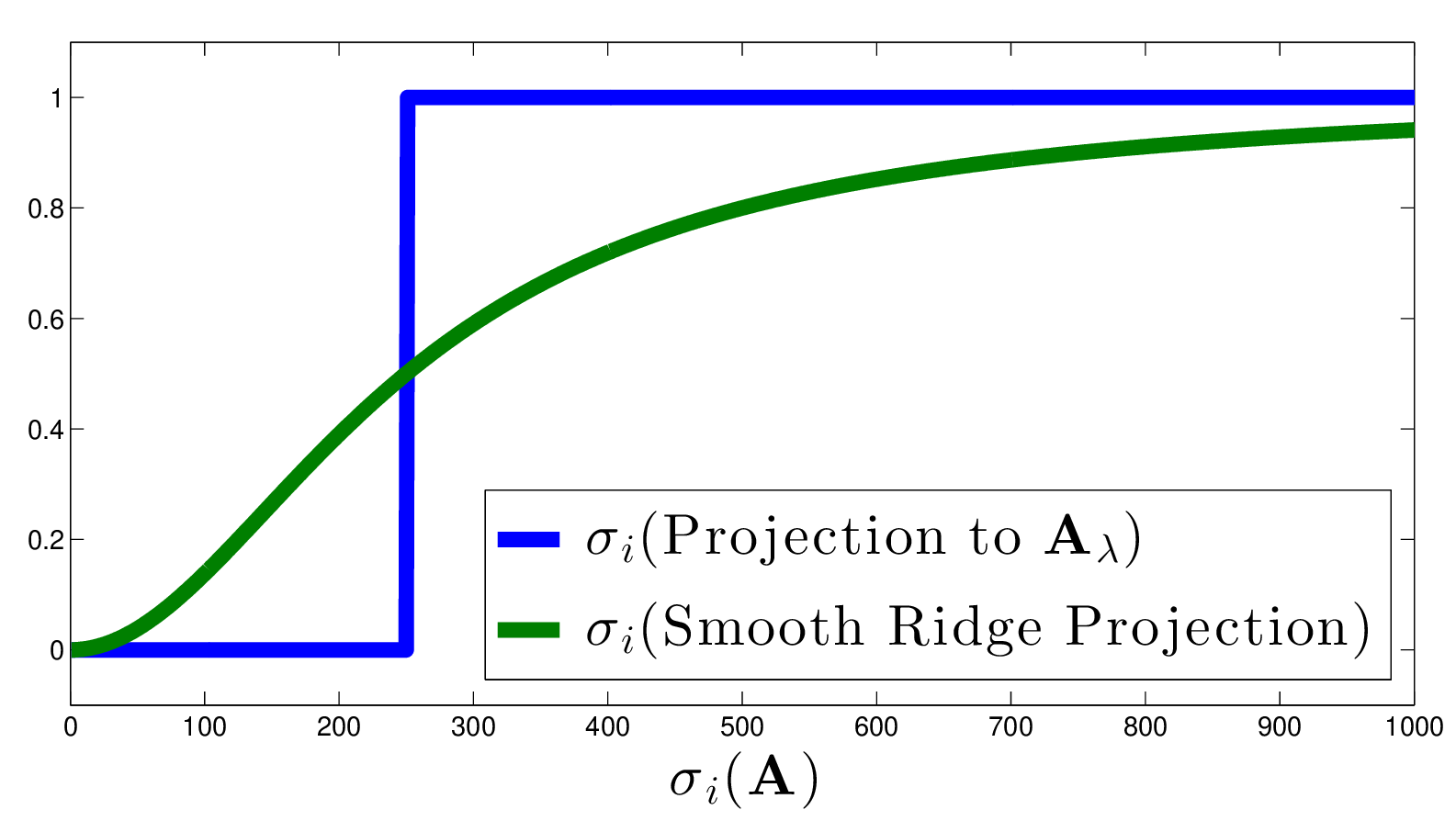}}
\fi
\caption{Singular values of the projection matrix $\bv{P}_{\bv{A}_\lambda}$ vs.\ those of the smooth projection operator $(\bv{A}^\T\bv{A}+\lambda \bv{I})^{-1}\bv{A}^\T\bv{A}$ obtained from ridge regression.}
\label{soft_projection}
\end{center}
\vskip -0.2in
\end{figure} 

Unfortunately, ridge regression is a very crude approximation to
PCR and projection in many settings and may perform significantly
worse in certain data analysis applications \cite{dhillon2013risk}. In
short, while ridge regression algorithms are valuable tools, it has
been unclear how to wield them for tasks like projection or PCR.

\subsection{Main result: from ridge regression to projection}

We show that it is possible to \emph{sharpen} the weak approximation given by ridge regression.
Specifically, there exists a low degree
polynomial $p(\cdot)$ such that $p\left((\bv{A}^\T\bv{A}+\lambda \bv{I})^{-1}\bv{A}^\T\bv{A} \right ) \bv{y}$ provides a very accurate approximation to $\bv{P}_{\bv{A}_\lambda} \bv{y}$ for any vector $\bv{y}$.
Moreover, the polynomial can be evaluated as a recurrence,
which translates into a simple iterative algorithm: we can apply the sharpened 
approximation to a vector by repeatedly applying
any ridge regression routine a small number of times. 

\begin{theorem}[Principal component projection without PCA]
  \label{thm:main}
Given $\bv{A}\in\mathbb{R}^{n \times d}$ and $\bv{y} \in \mathbb{R}^d$, Algorithm~\ref{alg:pcp} uses $\tilde O(\gamma^{-2} \log(1/\epsilon))$ approximate applications of $(\bv{A}^\T\bv{A}+\lambda \bv{I})^{-1}$ and returns $\bv{x}$ with $\norm{\bv{x}-\bv{P}_{\bv{A}_\lambda}\bv{y}}_2 \le \epsilon\norm{\bv{y}}_2$.
\end{theorem}
Like most iterative PCA algorithms, our running time scales
inversely with $\gap$, the \emph{spectral gap}  around $\lambda$.
\ificml

\fi
Notably, \ificml\else it \fi does not depend
on the number of principal components in $\bv{A}_\lambda$,
a cost incurred by any method that applies the projection $\bv{P}_{\bv{A}_\lambda}$ directly, either by explicitly computing the top principal components of $\bv{A}$,  or even by just computing an orthogonal span for these components.

As mentioned, the above theorem also yields an algorithm for principal component \emph{regression} that computes $\bv{A}_\lambda^\pinv \bv{b}$ without finding $\bv{A}_\lambda$. We achieve this result by introducing a robust reduction from projection to PCR, that again relies on ridge regression as a computational primitive.
\begin{corollary}[Principal component regression without PCA]\label{cor:pcr}
Given $\bv{A}\in\mathbb{R}^{n\times d}$ and $\bv{b} \in \mathbb{R}^n$, Algorithm~\ref{alg:pcr} uses $\tilde O(\gamma^{-2} \log(1/\epsilon))$ approximate applications of $(\bv{A}^\T\bv{A}+\lambda \bv{I})^{-1}$ and returns $\bv x$ with $\norm{\bv{x}-\bv{A}_\lambda^{\pinv}\bv{b}}_{\bv{A}^\T\bv{A}} \le~\epsilon \norm{\bv{b}}_2$.
\end{corollary}
Corollary~\ref{cor:pcr} gives the first known algorithm
for PCR that avoids the cost of principal component analysis.

\subsection{Related work}
A number of papers attempt to alleviate the high cost of principal component analysis when
solving PCR. It has been shown that an approximation to
$\bv{A}_\lambda$ suffices for solving the regression problem \cite{chan1990computing,boutsidis2014faster}. 
Unfortunately, even the fastest approximations are much slower than routines for ridge regression and 
inherently incur a linear dependence on the number of principal components above $\lambda$.

More closely related to our approach is work on the \emph{matrix sign function}, an important operation in control theory, quantum chromodynamics, and scientific computing in general.
Approximating the sign function often involves matrix polynomials 
similar to our ``sharpening polynomial'' that converts ridge regression to principal component projection. Significant effort addresses Krylov methods for applying 
such operators without computing them explicitly \cite{van2002numerical,Frommer2008}. 

Our work differs from these methods in an important way: since we only assume access to an approximate ridge regression algorithm, it is essential that our sharpening step is robust to noise. Our iterative polynomial construction allows for a complete and rigorous noise analysis that is not available for Krylov methods, while at the same time eliminating space and post-processing costs.
\ificml

\fi
Iterative approximations to the matrix sign function have been proposed, but lack rigorous noise analysis \cite{Higham:2008:FM}.

\subsection{Paper layout}
\begin{description}
\item Section~\ref{sec:prelims}: Mathematical and algorithmic preliminaries.
\item Section~\ref{sec:ridge2project}: Develop a PCA-free algorithm for principal component projection based on a ridge regression subroutine.
\item Section~\ref{sec:project2regress}: Show how our approximate projection algorithm can be used to solve PCR, again without PCA.
\item Section~\ref{sec:matrix-step-function}: Detail our iterative approach to sharpening the smooth ridge regression projection towards true projection via a low degree sharpening polynomial.
\item Section~\ref{sec:experiments}: Empirical evaluation of our principal component projection and regression algorithms.
\end{description}

\section{Preliminaries}\label{sec:prelims}

\ificml
\else
\newcommand\prelimHeader{\paragraph}
\fi

\prelimHeader{Singular value decomposition.}
Any matrix $\bv{A} \in \R^{n \by d}$ of rank $r$ has a
singular value decomposition (SVD) $\bv{A} =
\bv{U}\bv{\Sigma}\bv{V}^\T$, where $\bv{U}\in \mathbb{R}^{n \by r}$
and $\bv{V} \in \mathbb{R}^{d \by r}$ both have orthonormal columns
and $\bv{\Sigma} \in \R^{r\times r}$ is a diagonal matrix. The columns of
$\bv{U}$ and $\bv{V}$ are the left and right singular vectors of
$\bv{A}$. Moreover, $\bv{\Sigma} =
\diag(\sigma_1(\bv{A}),...,\sigma_r(\bv{A)})$, where $\sigma_1(\bv{A})
\ge \sigma_2(\bv{A}) \ge ... \ge \sigma_r(\bv{A}) > 0$ are the
singular values of $\bv{A}$ in decreasing order. 

The columns of $\bv{V}$ are the eigenvectors of the covariance matrix
$\bv{A}^\T \bv{A}$, i.e.\ the \emph{principal components} of the data,
and the eigenvalues of the covariance matrix are the squares of the
singular values $\sigma_1, \dots, \sigma_r$.

\prelimHeader{Functions of matrices.}
If $f : \R \to \R$ is a scalar function and $\vS = \diag(s_1, \dots,
s_n)$ is a diagonal matrix, we define by $f(\vS) \defeq \diag(f(s_1),
\dots, f(s_n))$ the entrywise application of $f$ to the diagonal. For a non-diagonal matrix $\bv{A}$ with SVD $\bv{A} = \bv{U\Sigma V}^\T$ we define $f(\vA) \defeq
\vU f(\vSigma) \vV^\T$.

\prelimHeader{Matrix pseudoinverse.}
We define the \emph{pseudoinverse} of $\bv{A}$ as
$\bv{A}^\pinv = f(\bv{A})^\T$ where $f(x) = 1/x$. The
pseudoinverse is essential in the context of regression, as the vector
$\bv{A}^\pinv \bv{b}$ minimizes the squared error $\ltwo{\bv{A}\bv{x}
  - \bv{b}}^2$.

\prelimHeader{Principal component projection.}
Given a threshold $\lambda > 0$ let $k$ be the largest index
with $\sigma_k(\bv{A})^2 \ge \lambda$ and define:
\begin{align*}
  \bv{A}_\lambda &~\defeq~
    \bv{U} \diag(\sigma_1, \dots, \sigma_k, 0, \dots, 0) \bv{V}^\T.
\end{align*}
The matrix $\bv{A}_\lambda$ contains $\bv{A}$'s rows projected to the
span of all principal components having squared singular value at
least $\lambda$. We sometimes write $\bv{A}_\lambda = \bv{A}
\bv{P}_{\bv{A}_\lambda}$ where $\bv{P}_{\bv{A}_\lambda}\in \R^{d\times
  d}$ is the projection onto these top components. Here
$\bv{P}_{\bv{A}_\lambda} = f(\bv{A}^\T\bv{A})$ where $f(x)$ is a step
function: $0$ if $x < \lambda$ and $1$ if $x \ge \lambda$.

\prelimHeader{Miscellaneous notation.}
For any positive semidefinite $\bv{M},\bv{N} \in \mathbb{R}^{d \times d}$ we use $\bv{N} \preceq \bv{M}$ to denote that $\bv{M}-\bv{N}$ is positive semidefinite. For any $\bv{x}\in\R^d$, $\norm{\bv{x}}_\bv{M} \eqdef \sqrt{\bv{x}^\T\bv{M}\bv{x}}$.

\prelimHeader{Ridge regression.}
Ridge regression is the problem of computing, given a regularization
parameter $\lambda > 0$:
\begin{align}
  \ridgesol{\vx}{\lambda} &\eqdef
    \argmin_{\vx \in \R^d} \ltwo{\vA \vx - \vb}^2 + \lambda \ltwo{\vx}^2.
  \label{eq:ridge}
\end{align}
The solution to \eqref{eq:ridge} is given by $\bv{x}^\lambda = \left (\bv{A}^\T \bv{A} + \lambda \bv{I} \right )^{-1}\bv{A}^\T \bv{b}$. Applying the matrix $\left(\bv{A}^\T \bv{A} + \lambda \bv{I} \right )^{-1}$ to $\bv{A}^\T \bv{b}$ is equivalent to solving the convex minimization problem:
\begin{align*}
\bv{x}^\lambda &= \argmin_{\bv{x}\in\mathbb{R}^d} \bv{x}^\T\bv{A}^\T\bv{A}\bv{x} - 2\bv{y}^\T\bv{x} + \lambda \norm{\bv{x}}_2^2,
\end{align*}
for $\bv{y} = \bv{A}^\T \bv{b}$.
A vast literature studies solving problems of this form via (accelerated) gradient descent, stochastic variants, and random sketching \cite{nesterov1983method,nelson2013osnap,shalev2014accelerated,lin2014accelerated,frostig15unregularizing,cohen2015uniform}. 
We summarize a few, now standard, runtimes achievable by these iterative methods:
\begin{lemma}[Ridge regression runtimes]\label{ridgeRuntime}
Given $\bv{y} \in \R^{d}$ let $\bv{x}^* =
(\bv{A}^\T\bv{A}+\lambda\bv{I})^{-1}\bv{y}.$ 
There is an algorithm,
$\textsc{ridge}(\bv{A},\bv{\lambda},\bv{y},\epsilon)$ that, for any $\epsilon > 0$, returns
$\bv{\tilde x}$ such that
$$\norm{\bv{\tilde x} - \bv{x}^*}_{\bv{A}^\T\bv{A}+\lambda\bv{I}} \le
\epsilon\norm{\bv{y}}_{(\bv{A}^\T\bv{A}+\lambda\bv{I})^{-1}}.$$ It
runs in time $T_{\textsc{ridge}}(\bv{A},\bv{\lambda},\epsilon) =
O\left (\nnz(\bv{A})\sqrt{\kappa_\lambda}\cdot \log(1/\epsilon)\right
)$ where $\kappa_\lambda= \sigma_1^2(\bv{A})/\lambda$ is the condition
number of the regularized system and $\nnz(\bv{A})$ is the number of
nonzero entries in $\bv{A}$. There is a also stochastic algorithm that, for
any $\delta > 0$, gives the same guarantee with probability $1-\delta$
in time
$$T_{\textsc{ridge}}(\bv{A},\bv{\lambda},\epsilon,\delta) = O\left (\left(\nnz(\bv{A}) + d \sr(\bv{A}) \kappa_{\lambda}\right)\cdot \log(1/\delta\epsilon)\right ),$$
where
$\sr(\bv A) = \norm{\bv{A}}_F^2 /\norm{\bv{A}}_2^2$ is $\bv{A}$'s stable rank.
When $\nnz(\bv{A}) \geq d \sr (\bv{A}) \kappa_\lambda$ the runtime can be improved to
$$T_{\textsc{ridge}}(\bv{A},\bv{\lambda},\epsilon,\delta) = \tilde{O} (
\sqrt{\nnz(\bv{A}) \cdot d \sr (\bv{A}) \kappa_\lambda}
\cdot \log(1/\delta\epsilon) ),$$
where the $\tilde{O}$ hides a factor of  $ \log \left(\frac {\nnz(\bv A)} {d \sr (\bv{A})\kappa_\lambda} \right)$.
\end{lemma}

Typically, the regularized condition number $\kappa_\lambda$ will be significantly smaller than the full condition number of $\bv{A}^\T \bv{A}$.

\section{From ridge regression to principal component projection}\label{sec:ridge2project}

We now describe how to approximately apply $\bv{P}_{\bv{A}_\lambda}$
using any black-box ridge regression routine. The key idea is to first
compute a soft step function of $\bv{A}^\T\bv{A}$ via ridge
regression, and then to sharpen this step to approximate
$\bv{P}_{\bv{A}_\lambda}$.
\ificml

\fi
Let $\bv{B}\bv{x} = (\bv{A}^\T \bv{A} +\lambda \bv{I})^{-1}(\bv{A}^\T
\bv{A})\bv{x}$ be the result of applying ridge regression to
$(\bv{A}^\T\bv{A})\bv{x}$. In the language of functions of matrices,
we have $\bv{B} = r(\bv{A}^\T\bv{A})$, where
\begin{align*}
r(x) \eqdef \frac{x}{x+\lambda}.
\end{align*}
The function $r(x)$ is a smooth step about $\lambda$ (see Figure
\ref{soft_projection}). It primarily serves to map the eigenvalues of
$\bv{A}^\T \bv{A}$ to the range $[0,1]$, mapping those exceeding the
threshold $\lambda$ to a value above $1/2$ and the rest to a value
below $1/2$.
\ificml

\fi
To approximate the projection $\bv{P}_{\bv{A}_{\lambda}}$, it would
now suffice to apply a simple symmetric step function:
\begin{align*}
  \label{eq:sharpen_function}
 s(x) =
  \begin{cases}
    0 &\textif x < 1/2 \\
    1 &\textif x \ge 1/2
  \end{cases}
\end{align*}
It is easy to see that $s(\bv{B}) = s(r(\bv{A}^\T\bv{A})) =
\bv{P}_{\bv{A}_\lambda}$. For $x \ge \lambda$, $r(x) \ge 1/2$ and so
$s(r(x)) = 1$. Similarly for $x < \lambda$, $r(x) < 1/2$ and hence
$s(r(x)) = 0$. That is, the symmetric step function exactly converts
our smooth ridge regression step to the true projection operator.

\subsection{Polynomial approximation to the step function}

While computing $s(\bv B)$ directly is expensive, requiring the SVD of $\bv{B}$, we show how to approximate this function with a \emph{low-degree polynomial}. We also show how to apply this polynomial efficiently and stably using a simple iterative algorithm. Our main result, proven in Section \ref{sec:matrix-step-function}, is:

\begin{lemma}[Step function algorithm]
\label{lem:matrix-step-main}
Let $\bv{S} \in \R^{d \times d}$ be symmetric with every eigenvalue $\sigma$ satisfying $\sigma \in [0, 1]$ and $|\sigma - 1/2| \geq \gap$. Let $\algA$ denote a procedure that on $\bv x \in \R^d$ produces $\algA(\bv x)$ with $\|\algA(\bv x) - \bv{S} \bv x\|_2 = O(\epsilon^2 \gap^2) \|\bv x\|_2$. Given $\bv y \in \R^d$ set
$\bv{s}_0 := \algA(\bv y)$, $\bv{w}_0 := \bv{s}_0 - \frac{1}{2} \bv{y}$, and for $k \geq 0$ set 
\[
\bv{w}_{k + 1} := 4 \left(\frac{2k + 1}{2k+2}\right) \algA (\bv{w}_{k} - \algA(\bv{w}_{k}))
\]
and $\bv{s}_{k + 1} := \bv{s}_{k} + \bv{w}_{k + 1}$. If all arithmetic operations are performed with $\Omega(\log(d/\epsilon \gap))$ bits of precision then 
$\|\bv{s}_{q} - s(\bv{S}) \bv y\|_2 = O(\epsilon) \|\bv{y}\|_2$ for $q = \Theta(\gap^{-2} \log(1/\epsilon))$.
\end{lemma}
Note that the output $\bv{s}_q$ is an approximation to a $2q$ degree polynomial of $\bv{S}$ applied to $\bv{y}$. 
In Algorithm \ref{alg:pcp}, we give pseudocode for combining the procedure with ridge regression to solve principal component projection.
Set $\bv{S} = \bv{B}$ and let $\algA$ be an algorithm that approximately applies $\bv{B}$ to any $\bv{x}$ by applying approximate ridge regression to $\bv{A}^\T \bv{A}\bv{x}$. As long as $\bv{B}$ has no eigenvalues falling within $\gap$ of $1/2$, the lemma ensures $\norm{\bv{s}_q - \bv{P}_{\bv{A}_\lambda}\bv{y}}_2 = O(\epsilon) \norm{\bv{y}}_2$. This requires $\gap$ on order of the \emph{spectral gap}: $1 -\sigma_{k+1}^2(\bv{A})/\sigma_{k}^2(\bv{A})$, where $k$ is the largest index with $\sigma_k^2(\bv{A}) \ge \lambda$. 

\begin{algorithm}
\caption{(\textsc{pc-proj}) Principal component projection}
\label{alg:pcp}
{\bf input}: $\bv{A} \in \mathbb{R}^{n \times d}$, $\bv{y} \in \mathbb{R}^{d}$, error $\epsilon$, failure rate $\delta$, threshold $\lambda$, gap $\gap \in (0,1)$
\begin{algorithmic}
\STATE{$q := c_1 \gap^{-2} \log(1/\epsilon)$}
\STATE{$\epsilon' := c_2^{-1} \epsilon^2\gap^2/\sqrt{\kappa_\lambda}$,\hspace{1em} $\delta' := \delta/(2q)$}
\STATE{$\bv{s} := \textsc{ridge}(\bv{A},\lambda,\bv{A}^\T\bv{A}\bv{y},\epsilon',\delta')$}
\STATE{$\bv{w} := \bv{s}-\frac{1}{2}\bv{y}$}
\FOR{$k = 0,..., q-1$}
\STATE{$\bv{t} := \bv{w} - \textsc{ridge}(\bv{A},\lambda,\bv{A}^\T\bv{A}\bv{w},\epsilon',\delta')$}
\STATE{$\bv{w} := 4\left (\frac{2k+1}{2k+2}\right ) \textsc{ridge}(\bv{A},\lambda,\bv{A}^\T\bv{A}\bv{t},\epsilon',\delta')$}
\STATE{$\bv{s} := \bv{s} + \bv{w}$}
\ENDFOR \\
\STATE{\textbf{return} $\bv{s}$}
\end{algorithmic}
\end{algorithm}
\begin{theorem}\label{pcpAlgoThm} If $\frac{1}{1-4\gap} \sigma_{k+1}(\bv{A})^2 \le \lambda \le (1-4\gap) \sigma_{k}(\bv{A})^2$ and $c_1,c_2$ are sufficiently large constants, \textsc{pc-proj} (Algorithm \ref{alg:pcp}) returns $\bv{s}$ such that with probability $\ge 1-\delta$,
\begin{align*}
\norm{\bv{s} -\bv{P}_{\bv{A}_\lambda}\bv{y}}_2 \le \epsilon \norm{\bv{y}}_2.
\end{align*}
The algorithm requires $O(\gap^{-2} \log(1/\epsilon))$ ridge regression calls, each costing $T_{\textsc{ridge}}(\bv{A},\lambda,\epsilon',\delta')$. Lemma \ref{ridgeRuntime} yields total cost (with no failure probability)
\begin{align*}
O\left (\nnz(\bv A) \sqrt{\kappa_\lambda} \gap^{-2}\log(1/\epsilon)\log\left (\kappa_\lambda/(\epsilon\gap)\right) \right )
\end{align*}
or, via stochastic methods,
\begin{align*}
\tilde O \left ( \left ( \nnz(\bv{A} + d\sr(\bv{A})\kappa_\lambda \right )\gap^{-2} \log(1/\epsilon)\log(\kappa_\lambda/(\epsilon\gap\delta ))\right )
\end{align*}
with acceleration possible when $\nnz(\bv{A}) > d\sr(\bv{A})\kappa_\lambda$.
\end{theorem}

\begin{proof} We instantiate Lemma \ref{lem:matrix-step-main}. Let $\bv{S} = \bv{B} = (\bv{A}^\T \bv{A} + \lambda \bv{I})^{-1} \bv{A}^\T\bv{A}$. As discussed, $\bv{B} = r(\bv{A}^\T\bv{A})$ and hence all its eigenvalues fall in $[0,1]$. Specifically, $\sigma_i(\bv{B}) = \frac{\sigma_i(\bv{A})^2}{\sigma_i(\bv{A})^2 + \lambda}$. Now, $\sigma_k(\bv{B}) \ge \frac{\lambda/(1-4\gap)}{\lambda/(1-4\gap) + \lambda} = \frac{1}{2-4\gap} \ge \frac{1}{2} + \gap$ and similarly $\sigma_{k+1}(\bv{B}) \le \frac{\lambda(1-4\gap)}{\lambda(1-4\gap) + \lambda} = \frac{1-4\gap}{2-4\gap} \le \frac{1}{2}- \gap$, so all eigenvalues of $\bv{B}$ are at least $\gap$ far from $1/2$. 
\ificml

\fi
By Lemma \ref{ridgeRuntime}, for any $\bv{x}$, with probability $\ge 1-\delta'$:
\begin{align*}
\|\textsc{ridge}(\bv{A},&\lambda,\bv{A}^\T \bv{A} \bv{x},\epsilon',\delta') - \bv{B}\bv{x}\|_{\bv{A}^\T\bv{A}+\lambda\bv{I}}\\ &\le \epsilon' \norm{\bv{A}^\T\bv{A}\bv{x}}_{(\bv{A}^\T\bv{A}+\lambda\bv{I})^{-1}} \le \sigma_1(\bv{A})\epsilon' \norm{\bv{x}}_{2}.
\end{align*}
Since the minimum eigenvalue of $\bv{A}^\T\bv{A}+\lambda \bv{I}$ is $\lambda$:
\begin{align*}
\|&\textsc{ridge}(\bv{A},\lambda,\bv{A}^\T \bv{A} \bv{x},\epsilon',\delta') - \bv{B}\bv{x}\|_{2} \\&\le \frac{\sigma_1(\bv{A})}{\sqrt{\lambda}} \epsilon' \norm{\bv{x}}_{2}
\le \frac{\sqrt{\kappa_\lambda}\epsilon^2\gap^2}{c_2\sqrt{\kappa_\lambda}}\norm{\bv{x}}_{2} = O(\epsilon^2\gap^2)\norm{\bv{x}}_2.
\end{align*}
Applying the union bound over all $2q$ calls of $\textsc{ridge}$, this bound holds for all calls with probability $\ge 1-\delta'\cdot 2q = 1-\delta.$
\ificml

\fi
So, overall, by Lemma  \ref{lem:matrix-step-main}, with probability at least $1-\delta$, $\norm{\bv{s}-s(\bv{B})\bv{y}}_2 = O(\epsilon)\norm{\bv y}_2$. As discussed, $s(\bv{B}) = \bv{P}_{\bv{A}_\lambda}$. Adjusting constants on $\epsilon$ (via $c_1$ and $c_2$) completes the proof.
\end{proof}

Note that the runtime of Theorem \ref{pcpAlgoThm} includes a dependence on $\sqrt{\kappa_\lambda}$. In performing principal component projection, $\textsc{pc-proc}$ applies an \emph{asymmetric step function} to $\bv{A}^\T\bv{A}$. The optimal polynomial for approximating this step also has a $\sqrt{\kappa_\lambda}$ dependence \cite{eremenko2011polynomials}, showing that our reduction from projection to ridge regression is optimal in this regard.

\subsection{Choosing $\lambda$ and $\gap$}

Theorem \ref{pcpAlgoThm} requires $\frac{\sigma_{k+1}(\bv{A})^2}{1-4\gap} \le \lambda \le (1-4\gap) \sigma_k(\bv{A})^2.$ %
If $\lambda$ is chosen approximately equidistant from the two eigenvalues, we need $\gap = O (1 -\sigma_{k+1}^2(\bv{A})/\sigma^2_k(\bv{A}) )$.

In practice, however, it is unnecessary to explicitly specify $\gap$ or to choose $\lambda$ so precisely. With $q = O(\gap^{-2}\log(1/\epsilon))$ our projection will be approximately correct on all singular values outside the range $[(1-\gap)\lambda, (1+\gap)\lambda]$. If there are any ``intermediate'' singular values in this range, as shown in Section \ref{sec:matrix-step-function}, the approximate step function applied by Lemma \ref{lem:matrix-step-main} will map these values to $[0,1]$ via a monotonically increasing soft step. That is, Algorithm \ref{alg:pcp} gives a slightly softened projection -- removing any principal directions with value $< (1-\gap)\lambda$, keeping any with value $> (1+\gap)\lambda$ and partially projecting away any in between.

\section{From principal component projection to principal component regression}\label{sec:project2regress}
A major motivation for an efficient, PCA-free method for projecting a vector onto the span of top principal components is principal component regression (PCR). Recall that PCR solves the following problem:
\begin{align*}
  \bv{A}_\lambda^\pinv \bv{b} = \argmin_{\bv{x} \in \R^{d}} \ltwo{\bv{A}_\lambda \bv{x} - \bv{b}}^2.
\end{align*}
In \emph{exact} arithmetic, $\bv{A}_\lambda^\pinv \bv{b}$ is equal to $(\bv{A}^\T\bv{A})^{-1}\bv{P}_{\bv{A}_\lambda}\bv{A}^\T \bv{b}$.
This identity suggests a method for computing the solution to ridge regression without finding $\bv{A}_\lambda$ explicitly: first apply a principal component projection algorithm to $\bv{A}^\T \bv{b}$ and then solve a linear system to apply $(\bv{A}^\T\bv{A})^{-1}$.

Unfortunately, this approach is disastrously unstable, not only when $\bv{P}_{\bv{A}_\lambda}$ is applied approximately, but in any finite precision environment. Accordingly, we present a modified method for obtaining PCA-free regression from projection.

\subsection{Stable inversion via ridge regression}

Let $\bv{y} = \bv{P}_{\bv{A}_\lambda}\bv{A}^\T \bv{b}$ and suppose we
have some $\bv{\tilde y} \approx \bv{y}$ (e.g.\ obtained from
Algorithm \ref{alg:pcp}).  The issue with the first approach mentioned
is that since $(\bv{A}^\T\bv{A})^{-1}$ could have a very large maximum
eigenvalue, we cannot guarantee $(\bv{A}^\T\bv{A})^{-1}\bv{\tilde y}
\approx (\bv{A}^\T\bv{A})^{-1}\bv{y}$. On the other hand, applying the
ridge regression operator $(\bv{A}^\T\bv{A}+\lambda\bv{I})^{-1}$ to
$\bv{\tilde y}$ is much more stable since it has a maximum eigenvalue
of $1/\lambda$, so $(\bv{A}^\T\bv{A}+\lambda\bv{I})^{-1}\bv{\tilde y}$
will approximate $(\bv{A}^\T\bv{A}+\lambda\bv{I})^{-1}\bv{y}$ well.

In short, it is more \emph{stable} to apply $(\bv{A}^\T\bv{A}+\lambda
\bv I)^{-1}\bv{y} = f(\bv{A}^\T \bv{A})\bv{y}$, where $f(x) =
\frac{1}{x+\lambda}$, but the goal in PCR is to apply
$(\bv{A}^\T\bv{A})^{-1} = h(\bv{A}^\T\bv{A})$ where $h(x) = 1/x$.
So, in order to go from one function to the other, we use a correction
function $g(x) = \frac{x}{1 - \lambda x}$. By simple calculation,
\begin{align*}
\bv{A}_\lambda^{\pinv}\bv{b} = (\bv{A}^\T \bv{A})^{-1} \bv{y} = g((\bv{A}^\T\bv{A}+\lambda\bv{I})^{-1})\bv{y}.
\end{align*}
Additionally, we can stably approximate $g(x)$ with an iteratively
computed low degree polynomial! Specifically, we use a truncation of
the series $ g(x) = \sum_{i=1}^\infty \lambda^{i-1} x^i.  $ An exact
approximation to $g(x)$ would exactly apply $(\bv{A}^\T\bv{A})^{-1}$,
which as discussed, is unstable due to very large eigenvalues
(corresponding to small eigenvalues of $\bv{A}^\T\bv{A}$).  Our
approximation to $g(x)$ is accurate on the large eigenvalues of
$\bv{A}^\T \bv{A}$ but \emph{inaccurate} on the small
eigenvalues. This turns out to be the key to the stability of our
algorithm. By not ``fully inverting'' these eigenvalues, our
polynomial approximation avoids the instability of applying the true
inverse $(\bv{A}^\T\bv{A})^{-1}$.  We provide a complete error analysis in
Appendix \ref{sec:pcr-appendix}, the upshot of which is the following:
\begin{lemma}[PCR approximation algorithm]
\label{lem:pcr-main}
Let $\algA$ be a procedure that, given $\bv x \in \R^d$, produces $\algA(\bv x)$ with $\|\algA(\bv x) - (\bv{A}^\T\bv{A}+\lambda \bv{I})^{-1} \bv x\|_{\bv{A}^\T\bv{A}+\lambda \bv{I}} =O( \frac{\epsilon}{q^2\sigma_1(\bv{A})}) \|\bv x\|_{2}$. Let $\algB$ be a procedure that, given $\bv x \in \R^d$ produces $\algB(\bv x)$ with $\norm{\algB(\bv x) - \bv{P}_{\bv{A}_\lambda}\bv{x}}_2 = O(\frac{\epsilon}{q^2\sqrt{\kappa_\lambda}}) \norm{\bv{x}}_2 $. Given $\bv b \in \R^n$ set $\bv{s}_0 := \algB(\bv{A}^\T \bv{b})$ and $\bv{s}_1 := \algA(\bv{s}_0)$. For $k \ge 1$ set:
\begin{align*}
\bv{s}_{k+1} := \bv{s}_1 + \lambda \cdot \algA(\bv{s}_k).
\end{align*}
If all arithmetic operations are performed with $\Omega(\log(d/q\epsilon))$ bits of precision then 
$\|\bv{s}_{q} - \bv{A}_\lambda^\pinv \bv b\|_{\bv{A}^\T \bv{A}} = O(\epsilon) \|\bv{b}\|_2$ for $q = \Theta(\log(\kappa_\lambda/\epsilon))$.
\end{lemma}

We instantiate the iterative procedure above in Algorithm \ref{alg:pcr}. $\textsc{pc-proj}(\bv{A},\lambda,\bv{y},\gap, \epsilon,\delta)$ denotes a call to Algorithm \ref{alg:pcp}.

\begin{algorithm}
\caption{(\textsc{ridge-pcr}) Ridge regression-based PCR}
\label{alg:pcr}
{\bf input}: $\bv{A} \in \mathbb{R}^{n \times d}$, $\bv{b} \in \R^n$, error $\epsilon$, failure rate $\delta$, threshold $\lambda$, gap $\gap \in (0,1)$
\begin{algorithmic}
\STATE{$q := c_1 \log(\kappa_\lambda/\epsilon)$}
\STATE{$\epsilon ' := c_2^{-1}\epsilon/(q^2 \sqrt{\kappa_\lambda})$, \hspace{1em}$\delta' = \delta/2(q+1)$}
\STATE{$\bv{y} := \textsc{pc-proj}(\bv{A},\lambda,\bv{A}^\T\bv{b},\gap,\epsilon',\delta/2)$}
\STATE{$\bv{s}_0 := \textsc{ridge}(\bv{A},\lambda,\bv{y},\epsilon',\delta')$, $\bv{s} := \bv{s}_0$}
\FOR{$k = 1,..., q$}
\STATE{$\bv{s} := \bv{s}_0 + \lambda \cdot  \textsc{ridge}(\bv{A},\lambda,\bv{s},\epsilon',\delta')$}
\ENDFOR
\STATE{\textbf{return} $\bv{s}$}
\end{algorithmic}
\end{algorithm}
\begin{theorem}\label{pcrAlgoThm} If $\frac{1}{1-4\gap} \sigma_{k+1}(\bv{A})^2 \le \lambda \le (1-4\gap) \sigma_{k}(\bv{A})^2$ and $c_1,c_2$ are sufficiently large constants, \textsc{ridge-pcr} (Algorithm \ref{alg:pcr}) returns $\bv{s}$ such that with probability $\ge 1-\delta$,
\begin{align*}
\norm{\bv{s} -\bv{A}_\lambda^\pinv \bv{b}}_{\bv{A}^\T\bv{A}} \le \epsilon \norm{\bv{b}}_2.
\end{align*}
The algorithm makes one call to $\textsc{pc-proj}$ and $O(\log(\kappa_\lambda/\epsilon))$ calls to ridge regression, each of which costs $T_{\textsc{ridge}}(\bv{A},\lambda,\epsilon',\delta')$, so
Lemma \ref{ridgeRuntime} and Theorem \ref{pcpAlgoThm} imply a total runtime of
\begin{align*}
\tilde O (\nnz(\bv A) \sqrt{\kappa_\lambda}\gap^{-2}\log^2\left (\kappa_\lambda/(\epsilon\gap)\right)),
\end{align*}
where $\tilde O$ hides $\log\log (1/\epsilon)$, or, with stochastic methods,
\begin{align*}
\tilde O ((\nnz(\bv A) + d\sr(\bv{A})\kappa_\lambda)\gap^{-2}\log^2\left (\kappa_\lambda/(\epsilon\gap\delta)\right) ).
\end{align*}
\end{theorem}
\begin{proof} We apply Lemma \ref{lem:pcr-main}; $\algA$ is given by $\textsc{ridge}(\bv{A},\lambda,\bv{x},\epsilon',\delta')$. Since $\norm{(\bv{A}^\T\bv{A}+\lambda \bv{I})^{-1}}_2< 1/\lambda$, Lemma \ref{ridgeRuntime} states that with probability $1-\delta'$,
\begin{align*}
\|&\algA(\bv{x}) - (\bv{A}^\T\bv{A}+\lambda\bv{I})^{-1}
\bv{x}\|_{\bv{A}^\T\bv{A}+\lambda \bv{I}}\\ &\le
\epsilon'\norm{\bv{x}}_{(\bv{A}^\T\bv{A}+\lambda\bv{I})^{-1}} \le
\frac{c_2^{-1}\epsilon}{q^2\sqrt{\kappa_\lambda
    \lambda}}\norm{\bv{x}}_2 \le
\frac{c_2^{-1}\epsilon}{q^2\sigma_1(\bv{A})} \norm{\bv x}_2.
\end{align*}

Now, $\algB$ is given by $\textsc{pc-proj}(\bv{A},\lambda,\bv{x},\gap,\epsilon',\delta/2)$. With probability $1-\delta/2$, if $\frac{1}{1-4\gap} \sigma_{k+1}(\bv{A})^2 \le \lambda \le (1-4\gap) \sigma_{k}(\bv{A})^2$ then by Theorem \ref{pcpAlgoThm},
$
\norm{\algB(\bv{x}) - \bv{P}_{\bv{A}_\lambda}\bv{x}}_2 \le \epsilon'\norm{\bv x}_2 = \epsilon/(c_2 q^2\sqrt{\kappa_\lambda}).
$
Applying the union bound over $q+1$ calls to $\algA$ and a single call to $\algB$, these bounds hold on every call with probability $\ge 1-\delta$. Adjusting constants on $\epsilon$ (via $c_1$ and $c_2$) proves the theorem.
\end{proof}

\section{Approximating the matrix step function}
\label{sec:matrix-step-function}
We now return to proving our underlying result on iterative polynomial approximation of the matrix step function:

\begin{replemma}{lem:matrix-step-main}[Step function algorithm]
Let $\bv{S} \in \R^{d \times d}$ be symmetric with every eigenvalue $\sigma$ satisfying $\sigma \in [0, 1]$ and $|\sigma - 1/2| \geq \gap$. Let $\algA$ denote a procedure that on $\bv x \in \R^d$ produces $\algA(\bv x)$ with $\|\algA(\bv x) - \bv{S} \bv x\| = O(\epsilon^2 \gap^2) \|\bv x\|_2$. Given $\bv y \in \R^d$ set
$\bv{s}_0 := \algA(\bv y)$, $\bv{w}_0 := \bv{s}_0 - \frac{1}{2} \bv{y}$, and for $k \geq 0$ set 
\[
\bv{w}_{k + 1} := 4 \left(\frac{2k + 1}{2k+2}\right) \algA (\bv{w}_{k} - \algA(\bv{w}_{k}))
\]
and $\bv{s}_{k + 1} := \bv{s}_{k} + \bv{w}_{k + 1}$. If all arithmetic operations are performed with $\Omega(\log(d/\epsilon \gap))$ bits of precision and if $q = \Theta(\gap^{-2} \log(1/\epsilon))$ then 
$\|s_{q} - s(\bv{S}) \bv{y}\|_2 = O(\epsilon) \|\bv{y}\|_2$.
\end{replemma}

The derivation and proof of Lemma~\ref{lem:matrix-step-main} is split into 3 parts. In Section~\ref{sub:poly_approx_sign} we derive a simple low degree polynomial approximation to the sign function:
\[
\sgn(x) \defeq 
\begin{cases}
1 & \text{ if } x > 0 \\
0 & \text{ if } x = 0 \\
-1 & \text{ if } x < 0
\end{cases} 
\]
In Section~\ref{sub:iterative_sign} we show how this polynomial can be computed with a stable iterative procedure. In Section~\ref{sub:step_function} we use these pieces and the fact that the step function is simply a shifted and scaled sign function to prove Lemma~\ref{lem:matrix-step-main}.
Along the way we give complementary views of
Lemma~\ref{lem:matrix-step-main} and show that there exist more
efficient polynomial approximations. 

\subsection{Polynomial approximation to the sign function}
\label{sub:poly_approx_sign}

We show that for sufficiently large $k$, the following polynomial is
uniformly close to $\sgn(x)$ on $[-1, 1]$:
\[
p_k(x) \defeq \sum_{i = 0}^{k} \left( x (1 - x^2)^i\prod_{j =1}^i \frac{2j - 1}{2j} \right)
\]

The polynomial $p_k(x)$ can be derived in several ways. One follows
from observing that $\sgn(x)$ is odd and thereby $\sgn(x) / x = 1 /
|x|$ is even. So, a good polynomial approximation for $\sgn(x)$ should
be odd and, when divided by $x$, should be even (i.e.\ a function of
$x^2$). Specifically, given a polynomial approximation $q(x)$ to
$1/\sqrt{x}$ on the range $(0, 1]$ we can approximate $\sgn(x)$ using
  $x q(x^2)$. Choosing $q$ to be the $k$-th order Taylor approximation
  to $1/\sqrt{x}$ at $x = 1$ yields $p_k(x)$. With this insight we
  show that $p_k(x)$ converges to $\sgn(x)$.
\begin{lemma}
\label{lem:convergence-of_p_k}
$\sgn(x) = \lim_{k \rightarrow \infty} p_k(x)$ for all $x \in [-1, 1]$. 
\end{lemma}
\begin{proof}
Let $f(x) = x^{-1/2}$. By induction on $k$ it is straightforward to show that the $k$-th derivative of $f$ at $x > 0$ is 
\[
f^{(k)}(x) =  \left(-1\right)^k \cdot \left(x\right)^{-\frac{1+2k}{2}} \prod_{i =1}^k \frac{2i - 1}{2}~.
\]
Since $(-1)^i (x - 1)^i = (1 - x)^i$ we see that the degree $k$ Taylor approximation to $f(x)$ at $x = 1$ is therefore
\[
q_k(x) = \sum_{i = 0}^{k} \left( (1 - x)^i \cdot \frac{1}{i!}
 \prod_{j =1}^i \frac{2j - 1}{2}\right) = \sum_{i = 0}^{k} \left( (1 - x)^i \cdot
 \prod_{j =1}^i \frac{2j - 1}{2j}\right) 
~.
\]
Note that for $x,y \in [\epsilon, 1]$, the remainder $f^{(k)}(x) (1 -
y)^k / k!$ has absolute value at most $(1 - \epsilon)^k$. Therefore
the remainder converges to $0$ as $k \rightarrow \infty$ and the
Taylor approximation converges, i.e.\ $\lim_{k \rightarrow \infty}
q_k(x) = 1/\sqrt{x}$ for $x \in (0, 1]$. Since $p_k(x) = x \cdot
  q_k(x^2)$ we have $\lim_{k \rightarrow \infty} p_k(x) = x /
  \sqrt{x^2} = \sgn(x)$ for $x \neq 0$ with $x \in [-1, 1]$. Since
  $p_k(0) = 0 = \sgn(0)$, the result follows.
\end{proof}

Alternatively, to derive $p_k(x)$ we can consider $(1 - x^2)^k$, which is relatively large near $0$ and small on the rest of $[-1, 1]$. Integrating this function from $0$ to $x$ and normalizing yields a good step function. In Appendix~\ref{sec:step-appendix} we prove that:
\begin{lemma} 
\label{lem:integral_form_of_pk}
For all $x \in \mathbb{R}$
\[
p_k(x) = \frac{\int_0^x (1 - y^2)^k dy}{\int_0^1 (1 - y^2)^k dy}
~.
\]
\end{lemma}
Next, we bound the rate of convergence of $p_k(x)$ to $\sgn(x)$:
\begin{lemma}
\label{lem:quality_of_sign_poly}
For $k \geq 1$ if $x \in (0, 1]$ then $p_k(x) > 0$ and
\begin{equation}
\label{eq:quality_of_sign_poly}
\sgn(x) - (x \sqrt{k})^{-1}e^{-k x^2} \leq p_k(x) \leq \sgn(x) ~.
\end{equation}
If $x \in [-1, 0)$ then $p_k(x) < 0$ and 
\[
\sgn(x) \leq p_k(x) \leq \sgn(x) + (x\sqrt{k})^{-1}e^{-kx^2}~.
\]
\end{lemma}
\begin{proof}
The claim is trivial when $x = 0$. Since $p_k(x)$ is odd it suffices
to consider $x \in (0, 1]$. For such $x$, it is direct that $p_k(x) >
0$, and $p_k(x) \leq \sgn(x)$ follows from the observation that
$p_k(x)$ increases monotonically with $k$ and $\lim_{k \rightarrow
  \infty} p_k(x) = \sgn(x)$ by Lemma~\ref{lem:convergence-of_p_k}. All
that remains to show is the left-side inequality of
\eqref{eq:quality_of_sign_poly}. Using
Lemma~\ref{lem:convergence-of_p_k} again,
\begin{align*}
\sgn(x) - p_k(x)
&= \sum_{i = k + 1}^{\infty} \left( x (1 - x^2)^i \prod_{j =1}^i \frac{2i - 1}{2i}\right)
\\
&\leq  x (1 - x^2)^{k} \sum_{i = 0}^{\infty} \left ((1 - x^2)^i \prod_{j =1}^k \frac{2j - 1}{2j}\right ) ~.
\end{align*}
Now since $1 + x \leq e^x$ for all $x$ and $\sum_{i=1}^n \frac{1}{i} \geq \ln n$, we have
\[
\prod_{j=1}^k \frac{2j - 1}{2j} \leq \exp \left(\sum_{j =1}^k \frac{- 1}{2j}
\right)
\leq \exp\left(\frac{- \ln k}{2}\right)
= \frac{1}{\sqrt{k}} ~. 
\]
Combining with $\sum_{i = 0}^\infty (1 - x^2)^i = x^{-2}$ and again
that $1 + x \leq e^x$ proves the left hand side of
\eqref{eq:quality_of_sign_poly}.
\end{proof}

The lemma directly implies that $p_k(x)$ is a high quality approximation to $\sgn(x)$ for $x$ bounded away from $0$.

\begin{corollary}
\label{cor:crude-convergence-of_pq}
If $x \in [-1, 1]$, with $|x| \geq \alpha > 0$ and $k = \alpha^{-2}
\ln(1/\epsilon)$, then $|\sgn(x) - p_k(x)| \leq \epsilon$.
\end{corollary}

We conclude by noting that this proof in fact implies the existence of
a lower-degree polynomial approximation to $\sgn(x)$. Since the sum of coefficients in our expansion is small, we can replace each $(1 -
x^2)^q$ with Chebyshev polynomials of lower degree. In
Appendix~\ref{sec:step-appendix}, we prove:

\begin{lemma}
\label{lem:chebyshev-improved-expansion}
There exists an $O(\alpha^{-1} \log(1/\alpha \epsilon))$ degree polynomial $q(x)$ such that $|\sgn(x) - q(x)| \leq \epsilon$ for all $x \in [-1, 1]$ with $|x| \geq \alpha > 0$ .
\end{lemma}

Lemma \ref{lem:chebyshev-improved-expansion} achieves, up an additive
$\log(1/\alpha)/\alpha$, the optimal trade off between degree and
approximation of $\sgn(x)$ \cite{eremenko2007uniform}.  We have preliminary progress toward making this near-optimal polynomial algorithmic, a topic we leave to explore in future work.

\subsection{Stable iterative algorithm for the sign function}
\label{sub:iterative_sign}

We now provide an iterative algorithm for computing $p_k(x)$ that works when applied with limited precision. Our formula is obtained by considering each term of $p_k(x)$. Let
\[
t_k(x) \defeq x (1 - x^2)^k \prod_{j =1}^k \frac{2j - 1}{2j}~.
\]
Clearly $t_{k + 1}(x)  = t_k(x) (1 - x^2) (2k + 1) / (2k + 2)$ and therefore we can compute the $t_{k}$ iteratively. Since $p_k(x) = \sum_{i = 0}^{k} t_i(x)$ we can compute $p_k(x)$ iteratively as well. We show this procedure works when applied to matrices, even if all operations are performed with limited precision:

\begin{lemma}
\label{lem:stable-sign}
Let $\bv{B} \in \R^{d \times d}$ be symmetric with $\|\bv{B}\|_2 \leq 1$. Let $\algC$ be a procedure that given $\bv x \in \R^d$ produces $\algC(\bv x)$ with 
$
\| \algC(\bv x) - (\bv{I} - \bv{B}^2) \bv x\|_2 \leq \epsilon \|\bv x\|_2.
$
Given $\bv y \in \R^d$ suppose that we have $\bv{t}_0$ and $\bv{p}_0$ such that $\|\bv{t}_0 - \bv{B} \bv{y}\|_2 \leq \epsilon \|\bv y\|_2$ and  $\|\bv{p}_0 - \bv{B} \bv y\|_2 \leq \epsilon \|\bv y\|_2$. For all $k \geq 1$ set
\[
\bv{t}_{k + 1} := \left(\frac{2k + 1}{2k+ 2}\right) \algC(\bv{t}_k)
\enspace \text{ and } \enspace
\bv{p}_{k + 1} := \bv{p}_k + \bv{t}_{k + 1} ~.
\]
Then if arithmetic operations are carried out with $\Omega(\log(d/\epsilon))$ bits of precision we have for $1 \leq k \leq 1/(7\epsilon)$ 
\[
\| t_{k}(\bv{B}) \bv{y} - \bv{t}_{k} \|_2 \leq 7k\epsilon 
\enspace \text{ and } \enspace
\| p_{k}(\bv{B}) \bv y - \bv{p}_k \|_2 \leq 7k\epsilon ~.
\]
\end{lemma}
\begin{proof}
Let $\bv{t}_k^* \defeq t_k(\bv{B}) \bv{y}$, $p_k^* \defeq p_k(\bv{B}) \bv y$, and  $\bv{C} \defeq \bv{I} - \bv{B}^2$. Since $\bv{p}_0^* = \bv{t}_0^* = \bv{B} \bv{y}$ and $\|\bv{B}\|_2 \leq 1$ we see that even if $\bv{t}_0$ and $\bv{p}_0$ are truncated to the given bit precision we still have $\|\bv{t}_0 - \bv{t}_0^*\|_2 \leq \epsilon \|\bv{y}\|_2$ and $\|\bv{p}_0 - \bv{p}_0^*\|_2 \leq  \epsilon \|\bv y\|_2$.

Now suppose that $\|\bv{t}_k - \bv{t}_k^*\|_2 \leq \alpha\norm{\bv y}_2$ for some $\alpha \le 1$. Since $|t_k(x)| \leq |p_k(x)| \leq |\sgn(x)| \leq 1$ for $x \in [-1, 1]$ and $-\bv{I} \preceq \bv{B} \preceq \bv{I}$ we know that $\|\bv{t}_k^*\|_2 \leq \|\bv{y}\|_2$ and by reverse triangle inequality $\|\bv{t}_k\|_2 \leq (1 + \alpha) \|\bv{y}\|_2$. Using our assumption on $\algC$ and applying triangle inequality yields 
\begin{align*}
\| \algC(\bv{t}_k) - \bv{C} \bv{t}_k^* \|_2
&\leq
\| \algC(\bv{t}_k)  - \bv{C} \bv{t}_k\|_2 + \|\bv{C} (\bv{t}_k - \bv{t}_k^*) \|_2
\\
&\leq 
\epsilon \| \bv{t}_k\|_2 + \|\bv{C} \|_2 \cdot  \|(\bv{t}_k - \bv{t}_k^*) \|_2
\\
&\leq
(\epsilon (1 + \alpha) + \alpha) \|\bv y\|_2 
\leq (2\epsilon + \alpha) \|\bv y\|_2 ~.
\end{align*}
In the last line we used $\|\bv{C}\|_2 \leq 1$ since $\bv{0} \preceq \bv{B}^2 \preceq \bv{I}$. Again, by this fact we know that $\|\bv{C} \bv{t}_k^* \|_2 \leq \|\bv{y}\|_2$ and therefore again by reverse triangle inequality $\|\algC(\bv{t}_k)\|_2 \leq (1 + 2\epsilon + \alpha) \|\bv y\|_2$. Using $\algC(\bv{t}_k)$ to compute $\bv{t}_{k + 1}$ with bounded arithmetic precision will then introduce an additional additive error of $\epsilon (1 + 2\epsilon + \alpha) \|\bv y\|_2 \leq 4\epsilon \|\bv y\|_2$. Putting all this together we have that $\|\bv{t}_k^* - \bv{t}_k\|_2$ grows by at most an additive $6\epsilon\norm{\bv{y}}_2$ every time $k$ increases and by the same argument so does $\|\bv{p}_k - \bv{p}_k^*\|_2$. Including our initial error of $\epsilon$ on $\bv{t}_0$ and $\bv{p}_0$, we conclude that $\|\bv{t}_k^* - \bv{t}_k\|_2$ and $\|\bv{p}_k - \bv{p}_k^*\|_2$ are both bounded by $6k\epsilon + \epsilon \leq 7k\epsilon$.
\end{proof}

\subsection{Approximating the step function}
\label{sub:step_function}

We finally apply the results of Section~\ref{sub:poly_approx_sign} and Section~\ref{sub:iterative_sign} to approximate the step function and prove Lemma~\ref{lem:matrix-step-main}. We simply apply the fact that $s(x) = (1/2) (1 + \sgn(2x - 1))$ and perform further error analysis. We first use Lemma~\ref{lem:stable-sign} to show how to compute $(1/2) ( 1 + p_k(2x - 1))$. 

\begin{lemma}
\label{lem:poly-step-approx}
Let $\bv{S} \in \R^{d \times d}$ be symmetric with $\bv{0} \preceq \bv{S} \preceq \bv{I}$. Let $\algA$ be a procedure that on $\bv x \in \R^d$ produces $\algA(\bv x)$ with $\|\algA(\bv x) - \bv{S}\bv x\|_2 \leq \epsilon \|\bv x\|_2$. Given arbitrary $\bv y \in \R^d$ set $\bv{s}_0 := \algA(\bv y)$, $\bv{w}_0 := \bv{s}_0 - (1/2) \bv{y}$, and for all $k \geq 0$ set
\[
\bv{w}_{k + 1} := 4 \left(\frac{2k + 1}{2k + 2}\right) \algA(\bv{w}_k - \algA(\bv{w}_k))
\]
and $\bv{s}_{k + 1} := \bv{s}_k + \bv{w}_{k + 1}$. If arithmetic operations are performed with $\Omega(\log(d/\epsilon))$ bits of precision and $k = O(1/\epsilon)$ then  
$
\| 1/2(\bv{I} - p_k(2\bv{S} - \bv{I})) \bv y - \bv{s}_k \|_2 = O(k \epsilon) \|\bv y\|_2 
$.
\end{lemma}
\begin{proof}
Since $\bv{M} \defeq  \bv{I} - (2\bv{S} - \bv{I})^2
= 4 \bv{S} (\bv{I} - \bv{S})$ we see that $\bv{w}_{k}$ is the same as $(1/2) \bv{t}_{k}$ in Lemma~\ref{lem:stable-sign} with $\bv{B} = 2 \bv{S} - \bv{I}$ and $\algC(\bv x) = 4 \algA(\bv x - \algA(\bv x))$, and $\bv{s}_{k} = \sum_{i = 0}^{k} (1/2) \bv{t}_i + (1/2) \bv{b}$. Since multiplying by $1/2$ everywhere does not increase error and since $\|2 \bv{S} - \bv{I}\|_2 \leq 1$ we can invoke Lemma~\ref{lem:stable-sign} to yield the result provided we can show
$
\| 4 \algA(\bv{x} - \algA(\bv{x})) - \bv{M} \bv{x}\|_2 = O(\epsilon) \|\bv x\|_2 
$.
Computing $\algA(\bv{x})$ and subtracting from $\bv{x}$ introduces at most additive error $2\epsilon \|\bv{x}\|_2$ Consequently by the error guarantee of $\algA$, $\| 4 \algA(\algA(\bv x) - \bv x) - \bv{M} \bv x\|_2 =O(\epsilon) \|\bv x\|_2$ as desired. 
\end{proof}
Using Lemma~\ref{lem:poly-step-approx} and Corollary~\ref{cor:crude-convergence-of_pq} we finally have:
\begin{proof}[Proof of Lemma~\ref{lem:matrix-step-main}]
By assumption, $\bv{0} \preceq \bv{S} \preceq \bv{I}$ and $\epsilon \gap^2 q = O(1)$. 
Invoking Lemma~\ref{lem:poly-step-approx} with error $\epsilon' = \epsilon^2\gap^2$, letting $\bv{a}_q \defeq 1/2(\bv{I} - p_q(2\bv{S} - \bv{I})) \bv{y}$ we have
\begin{equation}
\label{eq:main-iter-eq-1}
\| \bv{a}_q - \bv{s}_q \|_2 = O(\gap^2 \epsilon^2 q) \|\bv y\|_2 = O(\epsilon) \|\bv y\|_2 ~.
\end{equation}
Now, since $s(\bv{S}) = 1/2(\bv{I} - \sgn(2\bv{S} - \bv{I}))$ and every eigenvalue of $2\bv{S} - \bv{I}$ is in $[\gap, 1]$, by assumption on $\bv{S}$ we can invoke  Corollary~\ref{cor:crude-convergence-of_pq} yielding
$\| \bv{a}_q - s(\bv{S}) \bv{y}\|_2
\leq \frac{1}{2} \| p_q(2\bv{S} - \bv{I}) - \sgn(2 \bv{S} - \bv{I}) \|_2\norm{\bv y}_2
\leq 2\epsilon \norm{\bv y}_2.
$
The result follows from combining with \eqref{eq:main-iter-eq-1} via triangle inequality.
\end{proof}

\section{Empirical evaluation}\label{sec:experiments}
We conclude with an empirical evaluation of \textsc{pc-proc} and
\textsc{ridge-pcr} (Algorithms \ref{alg:pcp} and \ref{alg:pcr}).
Since PCR has already been justified as a statistical technique, we
focus on showing that, with few iterations, the algorithm recovers
an accurate approximation to $\bv{A}_\lambda^\pinv\bv{b}$ and
$\bv{P}_{\bv{A}_\lambda}\bv{y}$.

We begin with synthetic data, which lets us control the spectral gap $\gamma$ that dominates our iteration bounds (see Theorem \ref{pcpAlgoThm}). 
Data is generated randomly by drawing top singular values uniformly from the range $[.5(1+\gamma), 1]$ and tail singular values from $[0, .5(1-\gamma)]$. $\lambda$ is set to $.5$ and $\bv{A}$ is formed via the SVD $\bv{U}\bs{\Sigma}\bv{V}^\T$ where $\bv{U}$ and $\bv{V}$ are random orthonormal matrices and $\bs{\Sigma}$ contains our random singular values. To model a typical PCR application, $\bv{b}$ is generated by adding noise to the response $\bv{A}\bv{x}$ of a random ``true'' $\bv{x}$ that correlates with $\bv{A}$'s top principal components. 
\begin{figure}[h]
\centering
\subfigure[\normalsize{Regression}]{
	\label{synth_reg}
	\ificml
\else
	\includegraphics[width=.42\columnwidth]{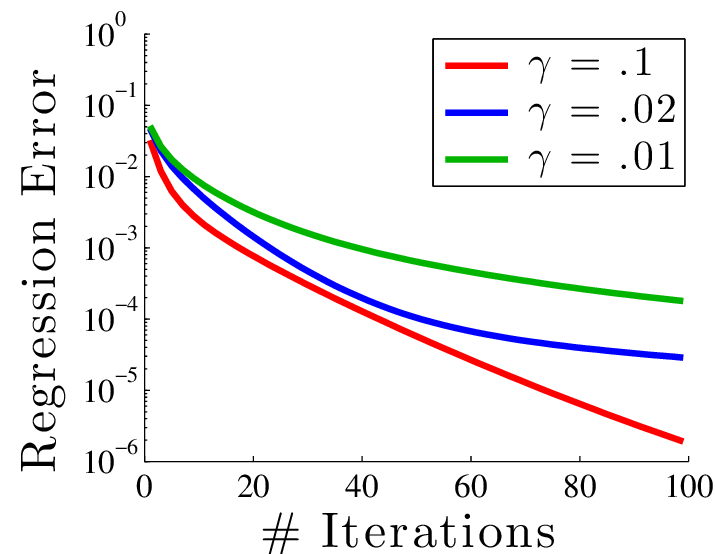}
	\fi
}
\subfigure[\normalsize{Projection}]{
	\label{synth_proj}
	\ificml
\else
	\hspace{2em}
	\includegraphics[width=.42\columnwidth]{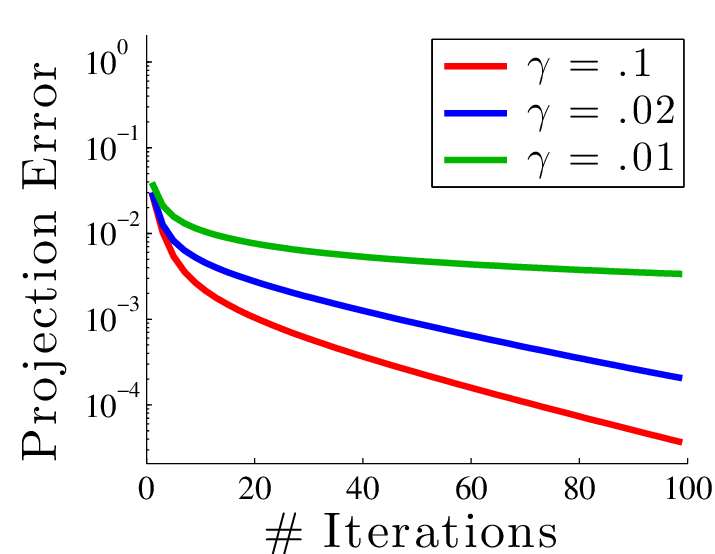}
	\fi
}
\caption{Relative error (shown on log scale) for \textsc{ridge-pcr} and \textsc{pc-proj} for synthetically generated data.}
\end{figure}

As apparent in Figure \ref{synth_reg}, our algorithm performs very well for regression, even for small $\gamma$. Error is measured via the natural $\bv{A}^\T\bv{A}$-norm and we plot
$\|\textsc{ridge-pcr}(\bv{A},\bv{b},\lambda) - \bv{A}_\lambda^\pinv\bv{b}\|_{\bv{A}^\T\bv{A}}^2/\|\bv{A}_\lambda^\pinv\bv{b}\|_{\bv{A}^\T\bv{A}}^2$.

Figure \ref{synth_proj} shows similar convergence for projection, although we do notice a stronger effect of a small gap $\gamma$ in this case.
Projection error is given with respect to the more natural 2-norm. 

Both plots confirm the linear convergence predicted by our analysis (Theorems \ref{pcpAlgoThm} and \ref{pcrAlgoThm}). 
To illustrate stability, we include an extended plot for the $\gamma=.1$ data which shows arbitrarily high accuracy as iterations increase (Figure \ref{synth_log}).
\begin{figure}[H]
\begin{center}
\ificml
\else
\centerline{\includegraphics[width=.55\columnwidth]{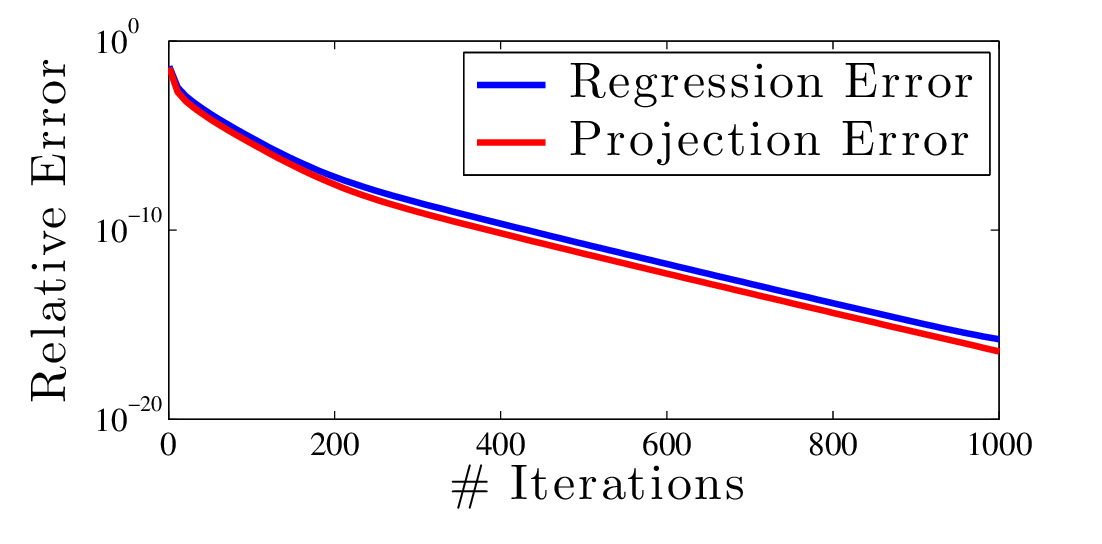}}
\fi
\caption{Extended log error plot on synthetic data with gap $\gamma = .1$.}
\label{synth_log}
\end{center}
\end{figure} 

Finally, we consider a large regression problem constructed from MNIST classification data \cite{mnist}, with the goal of distinguishing handwritten digits \{1,2,4,5,7\} from the rest.
Input is normalized and 1000 random Fourier features are generated according to a unit RBF kernel \cite{NIPS2007_3182}. Our final data set is both of larger scale and condition number than the original.
\begin{figure}[H]
\begin{center}
\ificml
\else
\centerline{\includegraphics[width=.55\columnwidth]{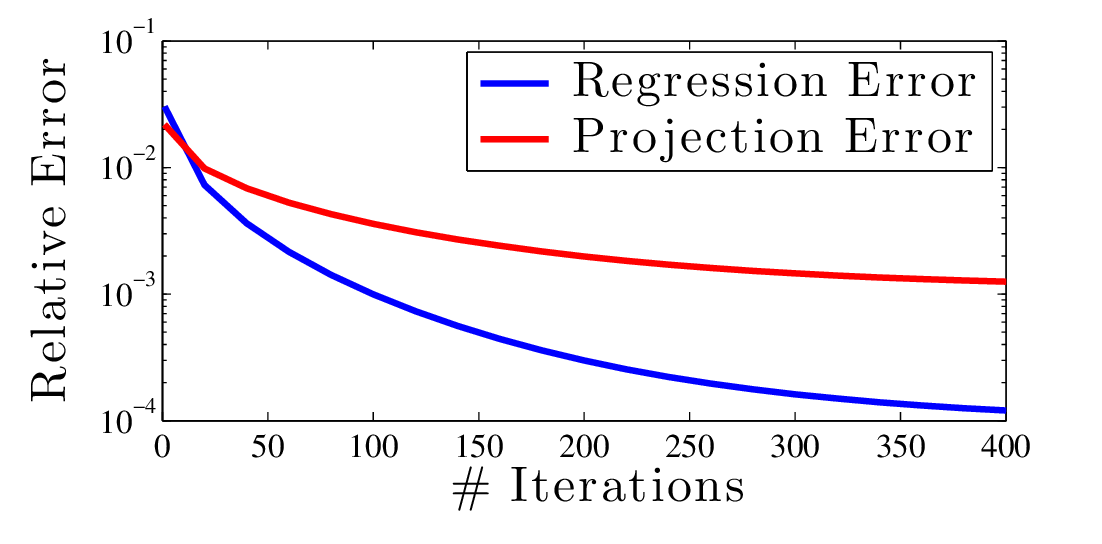}}
\fi
\caption{Relative error (on log scale) for \textsc{ridge-pcr} and \textsc{pc-proj} for an MNIST-based regression problem.}
\label{fig:mnist}
\end{center}
\end{figure} 
The MNIST principal component regression was run with $\lambda =
.01\sigma_1^2$. Although the gap $\gamma$ is very small around this
cutoff point (just $.006$), we see fast convergence for
PCR. Convergence for projection is slowed more notably by the small
gap, but it is still possible to obtain $0.01$ relative error with
only 20 iterations (i.e.\ invocations of ridge regression).

\bibliography{truncsys}
\ificml
\else
\bibliographystyle{alpha}
\fi
\appendix

\section{The matrix step function}
\label{sec:step-appendix}

Here we provide proofs omitted from Section~\ref{sec:matrix-step-function}. We prove Lemma~\ref{lem:integral_form_of_pk} showing that $p_q(x)$ can be viewed alternatively as a simple integral of $(1- x^2)^q$. We also prove Lemma~\ref{lem:chebyshev-improved-expansion} showing the existence of an even lower degree polynomial approximation to $\sgn(x)$.

\begin{replemma} {lem:integral_form_of_pk}
For all $x \in \mathbb{R}$
\[
p_k(x) = \frac{\int_0^x (1 - y^2)^k dy}{\int_0^1 (1 - y^2)^k dy}
~.
\]
\end{replemma}
\begin{proof}

Let $q_k(x) \defeq \int_{0}^{x} (1 - x^2)^q$. Our proof follows from simply recursively computing this integral via integration by parts.
Integration by parts 
with $u=(1-x^2)^k$ and $dv = dx$ yields
\[
q_k(x) 
= 
x(1-x^2)^k + 2k \int_{0}^{x} x^2(1-x^2)^{k - 1} ~.
\]
Since $x^2 = 1 - (1- x^2)$ we have
\[
q_k(x) = x(1-x^2)^k + 2k \cdot  q_{k - 1}(x) - 2k \cdot q_k(x).
\]
Rearranging terms and dividing by $2k + 1$ yields
\[
q_k(x) = \frac{1}{2k + 1} \left[x(1 - x^2)^k + 2k \cdot q_{k - 1}(x) \right].
\]
Since $q_0(1) = 1$ this implies that $q_k(1) = \prod_{j =1}^k \frac{2j}{2j + 1}$ and
\[
\frac{q_k(x)}{q_k(1)}
= \frac{1}{2k + 1} \prod_{j =1}^k \left(\frac{2j + 1}{2j}\right) x (1 - x^2)^k + \frac{q_{k - 1}(x)}{q_{k - 1}(1)} ~.
\]
Since $\frac{1}{2k + 1} \prod_{j =1}^k \frac{2j + 1}{2j} = \prod_{j =1}^k \frac{2j - 1}{2j}$ we have that $q_k(x) / q_k(1) = p_k(x)$ as desired.
\end{proof}

We now prove the existence of a lower degree polynomial for approximating $\sgn(x)$.
\begin{replemma}{lem:chebyshev-improved-expansion}
There exists an $O(\alpha^{-1} \log(1/\alpha \epsilon))$ degree polynomial $q(x)$ such that $|\sgn(x) - q(x)| \leq \epsilon$ for all $x \in [-1, 1]$ with $|x| \geq \alpha > 0$ .
\end{replemma}

We first provide a general result on approximating polynomials with lower degree polynomials.

\begin{lemma} [Polynomial Compression]
\label{lem:chebyshev-compression}
Let $p(x)$ be an $O(k)$ degree polynomial that we can write as
\[
p(x) = \sum_{i = 0}^{k} f_i(x) \left( g_i(x)\right)^i
\]
where $f_i(x)$ and $g_i(x)$ are $O(1)$ degree polynomials satisfying $|f_i(x)| \leq a_i$ and $|g_i(x)| \leq 1$ for all $x \in [-1, 1]$. Then, there exists polynomial $q(x)$ of degree $O(\sqrt{k \log(A/\epsilon))}$ where $A = \sum_{i = 0}^{k} a_i$ such that $|p(x) - q(x)| \leq \epsilon$ for all $x \in [-1, 1]$.
\end{lemma}

This lemma follows from the well known fact in approximation theory that there exist $O(\sqrt{d})$ degree polynomials that approximate $x^d$ uniformly on the interval $[-1, 1]$. In particular we make use of the following:

\begin{theorem}[Theorem~3.3 from \cite{sachdeva2014approx}]
\label{thm:monomial-approx}
For all $s$ and $d$ there exists a degree $d$ polynomial denoted $p_{s,d}(x)$ such that $|p_{s, d}(x) - x^s| \leq 2 \exp(-d^2/2s)$ for all $x \in [-1, 1]$.
\end{theorem}

Using Theorem~\ref{thm:monomial-approx} we prove Lemma~\ref{lem:chebyshev-compression}.

\begin{proof}
Let $d = \sqrt{2k \log(A/\epsilon)}$ and let our low degree polynomial be defined as
$
q(x) = \sum_{i =1}^k f_i(x) p_{d,i} (g_i(x))
$. By Theorem~\ref{thm:monomial-approx} we know that $q(x)$ has the desired degree and by triangle inequality for all $x \in [-1, 1]$
\begin{align*}
|p(x)- q(x)| 
&\leq \sum_{i =1}^k a_i |g_i(x)^i - p_{d,i}( g_{i}(x))|
\\
&\leq \sum_{i =1}^k a_i \exp(-d^2/2i) \leq \epsilon, 
\end{align*}
where the last line used $i \leq k$ and our choice of $d$.
\end{proof}

Using Lemma~\ref{lem:chebyshev-compression} we can now complete the proof.

\begin{proof}[Proof of Lemma~\ref{lem:chebyshev-improved-expansion}]
Note that $p_k(x)$ can be written in the form of Lemma~\ref{lem:chebyshev-compression} with $f_i(x) = x \prod_{j =1}^i \frac{2j - 1}{2j}$ and $g_i(x) = 1 - x^2$. Clearly $|f_i(x)| \leq 1$ and $|g_i(x)| \leq 1$ for $x \in [-1,1]$ and thus we can invoke $Lemma~\ref{lem:chebyshev-compression}$ to obtain a degree $O(\sqrt{k \log(k/\epsilon)})$ polynomial $q_k(x)$ with $|q_k(x) - p_k(x)| \leq \frac{1}{2} \epsilon$ for all $x \in [-1, 1]$. 

By Corollary~\ref{cor:crude-convergence-of_pq}  we know that for $k = \alpha^{-2} \ln(2/\epsilon)$ we have $|\sgn(x) - p_k(x)| \leq \epsilon/2$ and therefore $|\sgn(x) - q_k(x)| \leq \epsilon$. Since 
\[
\sqrt{\alpha^{-2} \ln(2/\epsilon) \ln\left(\frac{\alpha^{-2} \ln(2/\epsilon)}{\epsilon}\right)}
= O(\alpha^{-1} \ln (1/\alpha\epsilon)),
\]
we have the desired result.
\end{proof}

\section{Principal component regression}
\label{sec:pcr-appendix}
Finally we prove Lemma \ref{lem:pcr-main}, the main result behind our algorithm to convert principal component projection to PCR algorithm. The proof is in two parts. First, letting $\bv{y} = \bv{P}_{\bv{A}_\lambda}\bv{A}^\T \bv{b}$, we show how to approximate $(\bv{A}^\T \bv{A})^{-1}\bv{y} = \bv{A}_\lambda^{\pinv}$ with a low degree polynomial of the ridge inverse $(\bv{A}^\T \bv{A}+\lambda \bv{I})^{-1}$. Second, we provide an error analysis of our iterative method for computing this polynomial.

We start with a very basic polynomial approximation bound:
\begin{lemma}\label{pcrPoly} Let $g(x) \eqdef \frac{x}{1-\lambda x}$ and $p_k(x) \eqdef \sum_{i=1}^k \lambda^{i-1} x^i$. For $x \le \frac{1}{2\lambda}$ we have:
$
g(x) - p_k(x) \le \frac{1}{2^k\lambda}
$
\end{lemma}
\begin{proof}
We can expand $g(x) = \sum_{i=1}^\infty \lambda^{i-1}x^i$. So:
\begin{align*}
g(x) - p_k(x) = \sum_{i=k+1}^\infty \lambda^{i-1}x^i \le \frac{x}{2^{k-1}} \sum_{i=1}^\infty (\lambda x)^i \le \frac{1}{2^{k}\lambda}.
\end{align*}
\end{proof}
We next extend this lemma to the matrix case:

\begin{lemma}\label{pcrMatrixApprox} For any $\bv{A} \in \mathbb{R}^{n\times d}$ and $\bv{b} \in \mathbb{R}^d$, let $\bv{y} = \bv{P}_{\bv{A}_\lambda}\bv{A}^\T \bv{b}$. Let $p_k(x) = \sum_{i=1}^k \lambda^{i-1} x^i$. Then we have:
\begin{align*}
\norm{p_k\left ((\bv{A}^\T\bv{A}+\lambda\bv{I})^{-1}\right)\bv{y} - \bv{A}_\lambda^{\pinv}\bv{b}}_{\bv{A}^\T\bv{A}} \le \frac{\kappa_\lambda\norm{\bv{b}}_2}{2^{k}}.
\end{align*} 
\end{lemma}
\begin{proof}
For conciseness, in the remainder of this section we denote $\bv{M} \eqdef \bv{A}^\T\bv{A} +\lambda \bv{I}$. Let $\bv{z} = p_k\left (\bv{M}^{-1}\right)\bv{y}$. Letting $g(x) = x/(1-\lambda x)$ we have $g \left (\frac{1}{x + \lambda}\right ) = \frac{1}{x}$. So, $g(\bv{M}^{-1})\bv{y} = (\bv{A}^\T \bv{A})^{-1} \bv{y} = \bv{A}_\lambda^\pinv \bv{b}$. Define $\delta_k(x) \eqdef g(x)-p_k(x)$.
\begin{align}
\norm{\bv{A}_\lambda^{\pinv}\bv{b} - \bv{z}}_{\bv{A}^\T\bv{A}} &= \norm{\delta_k(\bv{M}^{-1})\bv{y}}_{\bv{A}^\T\bv{A}}\nonumber\\
&\le \sigma_1(\bv{A}) \cdot \norm{\delta_k(\bv{M}^{-1})\bv{y}}_{2}.\label{intialDeltaBound}
\end{align}

The projection $\bv{y}$ falls entirely in the span of principal components of $\bv{A}$ with squared singular values $\ge \lambda$. $\bv{M}$ maps these values to singular values $\le \frac{1}{\lambda + \lambda} = \frac{1}{2\lambda}$, and hence by Lemma \ref{pcrPoly} we have:
\begin{align*}
\norm{\delta_k(\bv{M}^{-1})\bv{y}}_{2} &\le \frac{1}{2^k\lambda} \norm{\bv{y}}_2\\
&\le \frac{\sigma_1(\bv{A})}{2^k\lambda}  \norm{\bv b}_2.
\end{align*}
Combining with \eqref{intialDeltaBound} and recalling that $\kappa_\lambda \eqdef \sigma_1(\bv{A})^2/\lambda$ gives the lemma.
\end{proof}

With this bound in place, we are ready to give a full error analysis of our iterative method for applying $p_k((\bv{A}^\T\bv{A}+\lambda\bv{I})^{-1})$.

\begin{replemma}{lem:pcr-main}[PCR approximation algorithm]
Let $\algA$ be a procedure that, given $\bv x \in \R^d$ produces $\algA(\bv x)$ with $\|\algA(\bv x) - (\bv{A}^\T\bv{A}+\lambda \bv{I})^{-1} \bv x\|_{\bv{A}^\T\bv{A}+\lambda \bv{I}} = O(\frac{\epsilon}{q^2\sigma_1(\bv{A})}) \|\bv x\|_{2}$. Let $\algB$ be a procedure that, given $\bv x \in \R^d$ produces $\algB(\bv x)$ with $\norm{\algB(\bv x) - \bv{P}_{\bv{A}_\lambda}\bv{x}}_2 = O(\frac{\epsilon}{q^2\sqrt{\kappa_\lambda}}) \norm{\bv{x}}_2 $. Given $\bv b \in \R^n$ set $\bv{s}_0 := \algB(\bv{A}^\T \bv{b})$ and $\bv{s}_1 := \algA(\bv{s}_0)$. For $k \ge 1$ set:
\begin{align*}
\bv{s}_{k+1} := \bv{s}_1 + \lambda \cdot \algA(\bv{s}_k)
\end{align*}
If all arithmetic operations are performed with $\Omega(\log(d/q\epsilon))$ bits of precision then 
$\|\bv{s}_{q} - \bv{A}_\lambda^\pinv \bv b\|_{\bv{A}^\T \bv{A}} = O(\epsilon) \|\bv{b}\|_2$ for $q = \Theta(\log(\kappa_\lambda/\epsilon))$.
\end{replemma}

\begin{proof}
Let $\bv{s}_0^* \eqdef \bv{P}_{\bv{A}_\lambda} \bv{A}^\T \bv{b}$, and for $k \ge 1$
\begin{align*}
\bv{s}_k^* = \sum_{i=1}^k \lambda^{i-1} \bv{M}^{-i} \bv{s}_0^*
\end{align*}

For ease of exposition, assume our accuracy bound on $\algB$ gives $\norm{\bv{s}_0-\bv{s}^*_0}_2 \le \frac{\epsilon}{q^2\sqrt{\kappa_\lambda}}\norm{\bv{A}^\T\bv{b}}_2 \le \sqrt{\lambda}\epsilon/q^2\norm{\bv{b}}_2$. Adjusting constants, the same proof with give Lemma \ref{lem:pcr-main} when error is actually $O(\sqrt{\lambda}\epsilon/q^2)$. By  triangle inequality:
\begin{align}\label{s1split}
&\norm{\bv{s}_1-\bv{s}_1^*}_{\bv{M}} \nonumber\\
&\le \| \bv{s}_1-\bv{M}^{-1}\bv{s}_0\|_{\bv M}+ \norm{\bv{M}^{-1}(\bv{s}_0^*-\bv{s}_0)}_{\bv{M}}
\end{align}
$\norm{\bv{M}^{-1}(\bv{s}_0^*-\bv{s}_0)}_{\bv{M}} \le \frac{1}{\sqrt{\lambda}}\norm{\bv{s}_0^*-\bv{s}_0}_2 \le \epsilon/q^2\norm{\bv{b}}_2$. And by our accuracy bound on $\algA$:
\begin{align*}
\| \bv{s}_1-\bv{M}^{-1}\bv{s}_0\|_{\bv{M}} \le \frac{\epsilon}{q^2\sigma_1(\bv{A})}\norm{\bv{s}_0}_{2}.
\end{align*}
Applying triangle inequality and the fact that the projection $\bv{P}_{\bv{A}_\lambda}$ can only decrease norm we have:
\begin{align*}
\frac{\epsilon}{q^2\sigma_1(\bv{A})}\norm{\bv{s}_0}_{2} &\le \frac{\epsilon}{q^2\sigma_1(\bv{A})}\left ( \norm{\bv{s}_0^*}_2 + \norm{\bv{s}_0^* - \bv{s}_0}_{2}  \right )\\
&\le  \frac{\epsilon}{q^2\sigma_1(\bv{A})}\left ( \norm{\bv{A}^\T \bv{b}}_2 + \sqrt{\lambda}\epsilon/q^2 \norm{\bv{b}}_2\right )\\
&\le 2\epsilon/q^2 \norm{\bv{b}}_2.
\end{align*}
Plugging back into \eqref{s1split} we finally have: $\norm{\bv{s}_1-\bv{s}_1^*}_{\bv{M}}  \le 3\epsilon/q^2\norm{\bv{b}}_2$.

Suppose we have for any $k \ge 1$, $\norm{\bv{s}_k-\bv{s}_k^*}_{\bv{M}}  \le \alpha\norm{\bv{b}}_2$.
\begin{align*}
\norm{\bv{s}_{k+1}-\bv{s}_{k+1}^*}_{\bv{M}} &\le \|\bv{s}_1-\bv{s}_1^*\|_{\bv{M}} + \lambda \norm{\algA(\bv{s}_k) - \bv{M}^{-1}\bv{s}_k^*}_{\bv{M}}\\
&\le 3\epsilon/q^2\norm{\bv{b}}_2 + \lambda \norm{\algA(\bv{s}_k) - \bv{M}^{-1}\bv{s}_k^*}_{\bv{M}}.
\end{align*}
We have:
\begin{align*}
&\lambda \|\algA(\bv{s}_k) - \bv{M}^{-1}\bv{s}_k^*\|_{\bv{M}} \\
&\le \frac{\lambda \epsilon}{q^2\sigma_1(\bv{A})} \norm{\bv{s}_k}_2 + \lambda \norm{\bv{M}^{-1}(\bv{s}_k - \bv{s}_k^*)}_{\bv{M}}\\
&\le \frac{\lambda \epsilon}{q^2\sigma_1(\bv{A})} \left (\norm{\bv{s}^*_k}_2 + \norm{\bv{s}_k -\bv{s}^*_k}_2 \right ) + \alpha\norm{\bv b}_2\\
&\le \frac{\lambda \epsilon}{q^2\sigma_1(\bv{A})} \norm{\bv{s}_k -\bv{s}^*_k}_2+ (\alpha+k\epsilon/q^2)\norm{\bv b}_2
\end{align*}
where the last step follows from: $ \frac{\lambda \epsilon}{q^2\sigma_1(\bv{A})} \norm{\bv{s}_k^*}_2 \le  \frac{\lambda \epsilon}{q^2\sigma_1(\bv{A})}\sum_{i=1}^k \lambda^{i-1} \norm{\bv{M}^{-i} \bv{s}_0^*}_2 \le \frac{k\epsilon}{q^2\sigma_1(\bv{A})} \norm{\bv{s}_0^*}_2 \le k\epsilon/q^2 \norm{\bv{b}}_2$. 

Now, $\frac{\lambda \epsilon}{q^2\sigma_1(\bv{A})}\norm{\bv{s}_k -\bv{s}^*_k}_2 \le \frac{ \epsilon\sqrt{\lambda}}{q^2\sigma_1(\bv{A})}\norm{\bv{s}_k -\bv{s}^*_k}_{\bv{M}} \le \epsilon\alpha/q^2\norm{\bv{b}}_2$ since $\sqrt{\lambda} \le \sigma_1(\bv{A})$.

So overall, presuming $\norm{\bv{s}_k-\bv{s}_k^*}_{\bv{M}}  \le \alpha\norm{\bv{b}}_2$, we have $\norm{\bv{s}_{k+1}-\bv{s}_{k+1}^*}_{\bv{M}} \le [(1+\epsilon/q^2)\alpha + (3+ k)\epsilon/q^2] \norm{\bv b}_2$. We know that $\norm{\bv{s}_1-\bv{s}_1^*}_{\bv{M}}  \le 3\epsilon/q^2\norm{\bv{b}}_2$, so by induction we have:
\begin{align*}
\norm{\bv{s}_{q}-\bv{s}_{q}^*}_{\bv{M}} < (1+\epsilon/q^2)^q \cdot q^2 \epsilon/q^2 \norm{\bv b}_2 < 3\epsilon \norm{\bv b}_2.
\end{align*}

Finally, applying the above bound, triangle inequality, and the polynomial approximation bound from Lemma \ref{pcrMatrixApprox} we have:
\begin{align*}
\|\bv{s}_{q} - \bv{A}_\lambda^\pinv \bv b\|_{\bv{A}^\T \bv{A}} &\le \norm{\bv{s}_{q}-\bv{s}_{q}^*}_{\bv{A}^\T\bv{A}} + \norm{\bv{s}_{q}^*-\bv{A}_\lambda^\pinv \bv{b}}_{\bv{A}^\T \bv{A}}\\
&\le \left (3\epsilon + \frac{\kappa_\lambda}{2^q} \right ) \norm{\bv b}_2 \le 4\epsilon \norm{\bv b}_2
\end{align*}
since $q = \Theta(\log(\kappa_\lambda/\epsilon))$.
\end{proof}

\end{document}